\documentclass{AHPCC-techreport}
\usepackage{latexsym}
\usepackage{epsf}
\usepackage{pstricks}
\usepackage{proof}

\newtheorem{lemma}{Lemma}[section]
\newtheorem{definition}[lemma]{Definition}
\newtheorem{theorem}[lemma]{Theorem}

\newtheorem{proposition}[lemma]{Proposition}

\def\ATS{\mbox{$\cal AT\kern-2pt S$}}

\def\D{{\cal D}}

\def\Linimp{\mbox{$\kern1pt=\kern-3pt\circ\;$}}

\def\Band{\land}
\def\Bimp{\supset}

\def\ST{\mbox{\it ST}}

\def\SlsegNone{\mbox{\it SlsegNone}}
\def\SlsegSome{\mbox{\it SlsegSome}}

\def\Timp{\Rightarrow}

\def\VTs{\overline{V\kern-1pt T}}

\def\Vand{\land}
\def\Vimp{\supset}
\def\Vimpzero{\Vimp_{0}}
\def\Vimpboth{\Vimp_{?}}
\def\VT{\mbox{\it V\kern-1pt T}}

\def\app{\mbox{\bf app}}
\def\atssv{\mbox{\it ATS/SV}}
\def\arrayView{\mbox{\it arrayView}}
\def\ArrayNone{\mbox{\it ArrayNone}}
\def\ArraySome{\mbox{\it ArraySome}}
\def\cI{\mbox{$c_I$}}
\def\cB{\mbox{$c_B$}}
\def\cond{\mbox{\bf if}}
\def\dc{\mbox{\it c}}
\def\dcc{\mbox{\it cc}}
\def\dcf{\mbox{\it cf\kern0.5pt}}
\def\dcfget{\mbox{\it get}}
\def\dcfgetPtr{\mbox{\it get\kern-0.5ptPtr}}
\def\dcfsetPtr{\mbox{\it set\kern-0.5ptPtr}}
\def\dcfgetPtrZ{\mbox{\it get\kern-0.5ptPtr}_0}
\def\dcfsetPtrZ{\mbox{\it set\kern-0.5ptPtr}_0}
\def\dcfgetRef{\mbox{\it get\kern-0.5ptRef}}
\def\dcfsetRef{\mbox{\it set\kern-0.5ptRef}}
\def\dcfmakePair{\mbox{\it makePair}}
\def\dom{\mbox{\bf dom}}
\def\dt{t}
\def\eDelta{\Delta^{l}}
\def\emptydctx{\emptyset}
\def\emptysctx{\emptyset}
\def\eforall{\forall^{-}}
\def\eguard#1{\supset^{-}\kern-2pt(#1)}
\def\erase#1{|#1|}
\def\falloc{\mbox{\it alloc}}
\def\ffalse{\mbox{\it false}}
\def\ffree{\mbox{\it free}}
\def\fix#1#2{\mbox{\bf fix}\;#1.#2}

\def\fread{\mbox{\it read}}
\def\fwrite{\mbox{\it write}}

\def\dcfgetFirst{\mbox{\it getFirst}}
\def\iforall{\forall^{+}}
\def\iguard#1{\supset^{+}\kern-2pt(#1)}

\def\linimp{\mbox{$-\kern-2pt\circ\;$}}
\def\lc{\mbox{\bf l}}
\def\letin#1#2{\mbox{\bf let}\;#1\;\mbox{\bf in}\;#2}
\def\lam#1#2{\mbox{\bf lam}\;#1.#2}
\def\lamsv{\lambda_{\it view}}
\def\lamsvplus{\lambda_{\it view}^{\forall,\exists}}
\def\pDelta{\Delta^{i}}
\def\pf{\mbox{\it pf\,}}
\def\pimp{\supset}
\def\pl{\underline{l}}
\def\prunit{{\mathbf 1}}
\def\pt{\underline{t}}
\def\pv{\underline{x}}
\def\sccnull{{\mathbf 0}}
\def\sccpair{\mbox{\it pair}}
\def\sccref{\mbox{\it ref}}

\def\scfsplit{\mbox{\it splitLemma}}
\def\scfunsplit{\mbox{\it unsplitLemma}}
\def\scfviewbox{\mbox{\it viewbox}}

\def\steval{\rightarrow_{\it ev/st}}

\def\saddr{\mbox{\it addr}}
\def\sbool{\mbox{\it bool}}
\def\sint{\mbox{\it int}}
\def\sllistView{\mbox{\it sllistView}}
\def\slsegView{\mbox{\it slsegView}}
\def\stype{\mbox{\it type}}
\def\sview{\mbox{\it view}}
\def\sviewtype{\mbox{\it viewtype}}
\def\subst#1#2#3{#3[#2\mapsto #1]}
\def\tBool{\mbox{\bf Bool}}
\def\tbool{\mbox{\bf bool}}
\def\temd{\models}
\def\timp{\rightarrow}
\def\timpzero{\rightarrow_{0}}
\def\timpboth{\rightarrow_{?}}
\def\tInt{\mbox{\bf Int}}
\def\tint{\mbox{\bf int}}
\def\tpjg{\vdash}
\def\tptr{\mbox{\bf ptr}}
\def\ttop{\mbox{\bf top}}
\def\ttrue{\mbox{\it true}}
\def\tunit{\mbox{\bf 1}}
\def\tuple#1{\langle #1\rangle}
\def\vB{\overline{B}}
\def\vimp{\mbox{$-\kern-2pt\circ\;$}}
\def\xf{\mbox{\it xf\/}}

\def\sqr#1#2{{\vcenter{\hrule height#2pt
     \hbox{\vrule width#2pt height#1pt \kern#1pt
        \vrule width#2pt}
     \hrule height#2pt}}}
\def\square{\kern2pt\raise0.25pt\hbox{$\mathchoice\sqr5{.6}\sqr5{.6}\sqr{4.1}{.5}\sqr{2.5}{.5}$}\kern2.5pt}

\def\fillsquare{\kern2pt\raise0.25pt
  \hbox{$\vcenter{\hrule height0pt \hbox{\vrule width5pt height5pt} \hrule height0pt}$}}
\newenvironment{proof}{{\bf Proof \/}\kern6pt}
{\unskip\nobreak\hfill\penalty50\kern4pt\hbox{}\nobreak\hfill\fillsquare\vskip6pt}

\title{To Memory Safety through Proofs\thanks{Partially supported by NSF grant
no. CCR-0229480}}

\author{%
Hongwei Xi and Dengping Zhu\\[4pt]
Computer Science Department\\[0pt]
Boston University\\[0pt]
}

\begin{document}

\maketitle

\keywords{Stateful View, Viewtype, Applied Type System, ATS}
\begin{abstract}
We present a type system capable of guaranteeing the memory safety of
programs that may involve (sophisticated) pointer manipulation such as
pointer arithmetic. With its root in a recently developed framework {\em
Applied Type System} ($\ATS$), the type system imposes a level of
abstraction on program states through a novel notion of recursive stateful
views and then relies on a form of linear logic to reason about such
stateful views.  We consider the design and then the formalization of the
type system to constitute the primary contribution of the paper. In
addition, we also mention a running implementation of the type system and
then give some examples in support of the practicality of programming with
recursive stateful views.
\end{abstract}

\section{Introduction}\label{section:introduction}
The need for direct memory manipulation through pointers is essential in
many applications and especially in those that involve systems
programming. However, it is also commonly understood that the use (or
probably misuse) of pointers is often a rich source for program errors.  In
safe programming languages such as ML and Java, it is completely forbidden
to make explicit use of pointers and memory manipulation is done through
systematic allocation and deallocation. In order to cope with applications
requiring direct memory manipulation, these languages often provide a means
to interface with functions written in unsafe languages such as C. While
this is a workable design, the evident irony of this design is that
probably the most difficult part of programming must be done in a highly
error-prone manner with little, if there is any, support of types. This
design certainly diminishes the efforts to promote the use of safe
languages such as ML and Java.

We have previously presented a framework {\em Applied Type System}
($\ATS$) to facilitate the design and formalization of type systems in
support of practical programming. It is already demonstrated that
various programming styles (e.g., modular programming~\cite{ModTyp},
object-oriented programming~\cite{GRDC03,OBJwMIpadl04},
meta-programming~\cite{MPicfp03}) can be directly supported {\em
  within} $\ATS$ without resorting to {\em ad hoc} extensions. In this
paper, we extend $\ATS$ with a novel notion of recursive stateful
views, presenting the design and then the formalization of a type
system $\atssv$ that can effectively support the use of types in
capturing program invariants in the presence of pointers.  For
instance, the interesting invariant can be readily captured in
$\atssv$ that each node in a doubly linked binary tree points to its
children that point back to the node itself, and this is convincingly
demonstrated in an implementation of AVL trees and splay
trees\cite{ats-lang}.  Also, we have presented previously a less
formal introduction to $\atssv$~\cite{zx-padl05}, where some short
examples involving stateful views can be found.

There are a variety of challenging issues that we must properly address in
order to effectively capture invariants in programs that may make
(sophisticated) use of pointers such as pointer arithmetic.  First and
foremost, we employ a notion of stateful views to model memory layouts.
For instance, given a type $T$ and an address $L$, we can form a
(primitive) stateful view $T@L$ to mean that a value of type $T$ is stored
at the address $L$. We can also form new stateful views in terms of
primitive stateful views. For instance, given types $T_1$ and $T_2$ and an
address $L$, we can form a view $(T_1@L)\otimes(T_2@L+1)$ to mean that a
value of type $T_1$ and another value of type $T_2$ are stored at addresses
$L$ and $L+1$, respectively, where $L+1$ stands for the address immediately
after $L$. Intuitively, a view is like a type, but it is linear.  Given a
term of some view $V$, we often say that the term proves the view $V$ and
thus refer to the term as a {\em proof} (of $V$).  It will soon become
clear that proofs of views cannot affect the dynamics of programs and thus
are all erased before program execution.

\begin{figure}[t]
\fontsize{11}{12}\selectfont
\begin{center}\begin{minipage}{14.2cm}
\begin{verbatim}
dataview arrayView (type, int, addr) =
  | {a:type, l:addr} ArrayNone (a, 0, l)
  | {a:type, n:int, l:addr | n >= 0}
      ArraySome (a, n+1, l) of (a @ l, arrayView (a, n, l+1))
\end{verbatim}
\end{minipage}\end{center}
\caption{An example of recursive stateful view}
\label{figure:arrayView}
\end{figure}
In order to model more sophisticated memory layouts, we need to form
recursive stateful views. For instance, we may use the concrete syntax in
Figure~\ref{figure:arrayView} to declare a (dependent) view constructor
$\arrayView$: Given a type $T$, an integer $I$ and an address $L$,
$\arrayView(T,I,L)$ forms a view stating that there are $I$ values of type
$T$ stored at addresses $L,L+1,\ldots,L+I-1$. There are two proof
constructors $\ArrayNone$ and $\ArraySome$ associated with $\arrayView$,
which are formally assigned the following views:
\begin{center}
\[\begin{array}{rcl}
\ArrayNone & : & \forall\lambda.\forall\tau.() \linimp \arrayView (\tau,0,\lambda) \\
\ArraySome & : & \forall\lambda.\forall\tau.\forall\iota.
\iota\geq 0\Bimp(\tau@\lambda\otimes\arrayView(\tau,\iota,\lambda+1)\linimp \arrayView(\tau,\iota+1,\lambda))
\end{array}\]
\end{center}
Note that we use $\otimes$ and $\linimp$ for linear (multiplicative)
conjunction and implication and $\tau$, $\iota$ and $\lambda$ for variables
ranging over types, integers and addresses, respectively.  Intuitively,
$\ArrayNone$ is a proof of $\arrayView(T,0,L)$ for any type $T$ and address
$L$, and $\ArraySome(\pf_1,\pf_2)$ is a proof of $\arrayView(T,I+1,L)$ for
any type $T$, integer $I$ and address $L$ if $\pf_1$ and $\pf_2$ are proofs
of views $T@L$ and $\arrayView(T,I,L+1)$, respectively.

Given a view $V$ and a type $T$, we can form a {\em viewtype} $V\Vand T$
such that a value of the type $V\Vand T$ is a pair $(\pf, v)$ in which
$\pf$ is a proof of $V$ and $v$ is a value of type $T$. For instance, the
following type can be assigned to a function $\fread_L$ that reads from the
address $L$:
\[\begin{array}{c}
(T@L)\Vand\tptr(L) \timp (T@L)\Vand T
\end{array}\]
Note that $\tptr(L)$ is the singleton type for the only pointer pointing to
the address $L$. When applied to a value $(\pf_1, L)$ of type
$(T@L)\Vand\tptr(L)$, the function $\fread_L$ returns a value $(\pf_2, v)$,
where $v$ is the value of type $T$ that is supposed to be stored at
$L$. Both $\pf_1$ and $\pf_2$ are proofs of the view $T@L$, and we may
think that the call to $\fread_L$ consumes $\pf_1$ and then produces
$\pf_2$. Similarly, the following type can be assigned to a function
$\fwrite_{L}$ that writes a value of type $T_2$ to the address $L$ where a
value of type $T_1$ is originally stored:
\[\begin{array}{c}
(T_1@L)\Vand(\tptr(L) * T_2) \timp (T_2@L)\Vand \tunit
\end{array}\]
Note that $\tunit$ stands for the unit type. In general, we can assign the
read ($\dcfgetPtr$) and write ($\dcfsetPtr$) functions the following types:
\[\begin{array}{rcl}
\dcfgetPtr &~~:~~&
\forall\tau.\forall\lambda.(\tau@\lambda)\Vand\tptr(\lambda)\timp(\tau@\lambda)\Vand\tau
\\
\dcfsetPtr &~~:~~&
\forall\tau_1.\forall\tau_2.\forall\lambda.(\tau_1@\lambda)\Vand(\tptr(\lambda) * \tau_2)\timp(\tau_2@\lambda)\Vand\tunit
\\
\end{array}\]

\begin{figure}
\fontsize{9}{10}\selectfont
\begin{verbatim}
fun getFirst {a:type, n:int, l:addr | n > 0} (pf: arrayView (a,n,l) | p: ptr(l))
  : '(arrayView (a,n,l) | a) =
  let
     prval ArraySome (pf1, pf2) = pf // pf1: a@l and pf2: arrayView (a,n-1,l+1)
     val '(pf1' | x) = getPtr (pf1 | p) // pf1: a@l
  in
     '(ArraySome (pf1', pf2) | x)
  end

// The following is a proof function and thus is required to be total
prfun splitLemma {a:type, n:int, i:int, l:addr | 0 <= i, i <= n} .<i>.
     (pf: arrayView (a, n, l)): '(arrayView (a, i, l), arrayView (a, n-i, l+i)) =
  sif i == 0 then '(ArrayNone (), pf) // [sif]: static [if] for forming proofs
  else
    let
       prval ArraySome (pf1, pf2) = pf // this cannot fail as [i > 0] holds
       prval '(pf21, pf22) = splitLemma {a,n-1,i-1,l+1} (pf2)
    in
       '(ArraySome (pf1, pf21), pf22)
    end

fun get {a:type, n:int, i:int, l:addr | 0 <= i, i < n}
   (pf: arrayView (a, n, l) | p: ptr l, offset: int i): '(arrayView (a, n, l) | a) =
  let
     // pf1: arrayView (a,i,l) and pf2: arrayView (a,n-i,l+i)
     prval '(pf1, pf2) = splitLemma {a,n,i,l} (pf)
     val '(pf2 | x) = getFirst (pf2 | p + offset)
  in
     '(unsplitLemma (pf1, pf2) | x)
  end

\end{verbatim}
\caption{A programming example involving recursive stateful views}
\label{figure:arraySubscript}
\end{figure}
In order to effectively support programming with recursive stateful views,
we adopt a recently proposed design that combines programming with theorem
proving~\cite{CPwTP}.  While it is beyond the scope of the paper to formally
explain what this design is, we can readily use some examples to provide
the reader with a brief overview as to how programs and proofs are combined
in this design.  Also, these examples are intended to provide the reader
with some concrete feel as to what can actually be accomplished in
$\atssv$.  Of course, we need a process to elaborate programs written in
the concrete syntax of ATS into the (kind of) formal syntax of $\lamsvplus$
(presented in Section~\ref{section:extension}). This is a rather involved
process, and we unfortunately could not formally describe it in this paper
and thus refer the interested reader to~\cite{ATSdml} for further
details. Instead, we are to provide some (informal) explanation to
facilitate the understanding of the concrete syntax we use.

We have so far finished a running implementation of
ATS~\cite{ats-lang}, a programming language with its type system
rooted in the framework $\ATS$, and $\atssv$ is a part of the type
system of ATS.  In Figure~\ref{figure:arraySubscript}, we present some
code in ATS.  We use $'(\ldots)$ to form tuples in the concrete
syntax, where the quote symbol ($'$) is solely for the purpose of
parsing. For instance, $'()$ stands for the unit (i.e., the tuple of
length 0).  Also, the bar symbol ($\mid$) is used as a separator (like
the comma symbol ($,$)).  It should not be difficult to relate the
concrete syntax to the formal syntax of $\lamsvplus$ introduced later
in Section~\ref{section:extension} (assuming that the reader is
familiar with the SML syntax).  The header of the function
$\dcfgetFirst$ in Figure~\ref{figure:arraySubscript} indicates that
the following type is assigned to it:
\begin{center}
\[\begin{array}{l}
\forall\tau.\forall\iota.\forall\lambda.
\iota>0 \Bimp (\arrayView(\tau,\iota,\lambda)\Vand\tptr(\lambda) \timp \arrayView(\tau,\iota,\lambda)\Vand\tau)
\end{array}\]
\end{center}
where $\iota>0$ is a guard to be explained later.  Intuitively, when applied to
a pointer that points to a nonempty array, $\dcfgetFirst$ takes out the
first element in the array.  In the body of $\dcfgetFirst$, $\pf$ is a
proof of the view $\arrayView(a,n,l)$, and it is guaranteed to be of the
form $\ArraySome(\pf_1,\pf_2)$, where $\pf_1$ and $\pf_2$ are proofs of
views $a@l$ and $\arrayView(a,n-1,l+1)$, respectively; thus $\pf\,'_1$ is
also a proof of $\tau@\lambda$ and $\ArraySome(\pf\,'_1,\pf_2)$ is a proof
of $\arrayView(a,n,l)$.  In the definition of $\dcfgetFirst$, we have both
code for dynamic computation and code for static manipulation of proofs of
views, and the latter is to be erased before dynamic computation
starts. For instance, the definition of $\dcfgetFirst$ turns into:
\begin{verbatim}
fun getFirst (p) = let val x = getPtr p in x end
\end{verbatim}
after the types and proofs in it are erased; so the function can
potentially be compiled into one load instruction after it is inlined.

We immediately encounter an interesting phenomenon when attempting to
implement the usual array subscripting function $\dcfget$ of the following
type:
\begin{center}
\[\begin{array}{l}
\forall\tau.\forall n:\sint.\forall i:\sint.\forall\lambda.\\
\kern12pt
0\leq i \land i<n\Bimp
(\arrayView(\tau,n,\lambda)\Vand(\tptr(\lambda)*\tint(i)) \timp \arrayView(\tau,n,\lambda)\Vand\tau)
\end{array}\]
\end{center}
where $\tint(I)$ is a singleton type for the integer equal to $I$.  This
type simply means that $\dcfget$ is expected to return a value of type $T$
when applied to a pointer and a natural number such that the pointer points
to an array whose size is greater than the natural number and each element
in the array is of type $T$. Obviously, for any $0\leq i \leq n$, an array
of size $n$ at address $L$ can be viewed as two arrays: one of size $i$ at
$L$ and the other of size $n-i$ at $L+i$. This is what we call {\em view
change}, which is often done implicitly and informally (and thus often
incorrectly) by a programmer. In Figure~\ref{figure:arraySubscript}, the
proof function $\scfsplit$ is assigned the following functional view:
\begin{center}
\[\begin{array}{l}
\forall\tau.\forall n:\sint.\forall i:\sint.\forall\lambda.\\
\kern12pt
0\leq i\land i\leq n \Bimp (\arrayView(\tau,n,\lambda) \linimp \arrayView(\tau,i,\lambda)\otimes\arrayView(\tau,n-i,\lambda+i)) \\
\end{array}\]
\end{center}
Note that $\scfsplit$ is recursively defined and the termination metric
$\tuple{i}$ is used to verify that $\scfsplit$ is terminating. Please
see~\cite{XiHOSC02} for details on such a termination verification
technique.  To show that $\scfsplit$ is a total function, we also need to
verify the following pattern matching in its body:
\begin{verbatim}
                  prval ArraySome (pf1, pf2) = pf
\end{verbatim}
can never fail. Similarly, we can also define a total function
$\scfunsplit$ that proves the following view:
\begin{center}
\[\begin{array}{l}
\forall\tau.\forall n:\sint.\forall i:\sint.\forall\lambda.\\
\kern12pt
0\leq i\land i\leq n\Bimp (\arrayView(\tau,i,\lambda)\otimes\arrayView(\tau,n-i,\lambda+i)\linimp \arrayView(\tau,n,\lambda)) \\
\end{array}\]
\end{center}
With both $\scfsplit$ and $\scfunsplit$ to support view changes, an
$O(1)$-time array subscripting function is implemented in
Figure~\ref{figure:arraySubscript}. The definition of $\dcfget$ turns into:
\begin{verbatim}
fun get (p, i) = let val x = getFirst (p + i) in x end
\end{verbatim}
after the types and proofs in it are erased; so the function can
potentially be compiled into one load instruction after it is inlined.

We organize the rest of the paper as follows. In
Section~\ref{section:formal_development}, we formalize a language $\lamsv$
in which views, types and viewtypes are all supported. We then briefly
mention in Section~\ref{section:extension} an extension $\lamsvplus$ of
$\lamsv$ in which we support dependent types as well as polymorphic types,
and we also present some examples to show how views can be used in
practical programming.  In Section~\ref{section:PSV}, we present an
overview of the notion of persistent stateful views, which is truly
indispensable in practical programming.  We use $\atssv$ essentially for
the type system that extends $\lamsvplus$ with persistent stateful views.
Lastly, we mention some related work and conclude.

\begin{figure}[thp]
\[\begin{array}{lrcl}
\mbox{addresses} & L & ::= & \lc_0 \mid \lc_1 \mid \ldots \\
\mbox{views} & V & ::= & T@L \mid \prunit \mid V_1\otimes V_2 \mid V_1\vimp V_2 \\
\mbox{proof terms} & \pt & ::= & \pv \mid \pl \mid \tuple{\pt_1,\pt_2} \mid
\letin{\tuple{\pv_1,\pv_2}=\pt_1}{\pt_2} \mid \lambda \pv.\pt \mid \pt_1(\pt_2) \\
\mbox{proof var. ctx.} & \Pi & ::= & \emptysctx \mid \Pi,a:V \\
\\
\mbox{types} & T & ::= & \tBool \mid \tInt \mid \tptr(L) \mid \tunit \mid
V\Vimp\VT \mid T*T \mid \VT\timp\VT \\
\mbox{viewtypes} & \VT & ::= & \tBool \mid \tInt \mid \tptr(L) \mid \tunit \mid V\Vand\VT \mid
V\Vimpzero\VT \mid V\Vimp\VT \mid \\
&&&\VT_1 * \VT_2 \mid  \VT\timpzero\VT \mid \VT\timp\VT \\
\mbox{dyn. terms} & \dt &::= &
x \mid f \mid \dcc(\dt_1,\ldots,\dt_n) \mid
\dcf(\dt_1,\ldots,\dt_n) \mid \\
& & & \fread (\pt,\dt) \mid \fwrite (\pt,\dt_1,\dt_2) \mid \\
& & & \cond(\dt_1,\dt_2,\dt_3) \mid \fread(\pt,\dt) \mid \fwrite(\pt,\dt_1,\dt_2) \mid \\
& & & \pt\Vand \dt \mid \letin {\pv\Vand x=\dt_1}{\dt_2} \mid \lambda\pv.v \mid \dt(\pt) \mid \\
& & & \tuple{} \mid \tuple{\dt_1,\dt_2} \mid \letin{\tuple{x_1,x_2}=\dt_1}{\dt_2} \mid \\
& & & \lam{x}{\dt} \mid \app(\dt_1,\dt_2) \mid \fix{f}{\dt} \\
\mbox{values} & v & ::= & x \mid \dcc(v_1,\ldots,v_n) \mid \pt\Vand v \mid \lambda\pv.v \mid \tuple{v_1,v_2} \mid \lam{x}{d} \\
\mbox{int. dyn. var. ctx.} & \pDelta & ::= & \emptydctx \mid \pDelta,x:T \\
\mbox{lin. dyn. var. ctx.} & \eDelta & ::= & \emptydctx \mid \eDelta,x:\VT \\
\mbox{dyn. var. ctx.} & \Delta & ::= & (\pDelta;\eDelta) \\
\\
\mbox{state types} & \mu & ::= & [] \mid \mu[l\mapsto T] \\
\mbox{states} & \ST & ::= & [] \mid \ST[l\mapsto v] \\
\end{array}\]
\caption{The syntax for $\lamsv$}
\label{figure:syntax}
\end{figure}
\section{Formal Development}\label{section:formal_development}
In this section, we formally present a language $\lamsv$ in which the type
system supports views, types and viewtypes. The main purpose of formalizing
$\lamsv$ is to allow for a gentle introduction to unfamiliar concepts such
as view and viewtype. To some extent, $\lamsv$ can be compared to the
simply typed lambda-calculus, which forms the core of more advanced typed
lambda-calculi.  We will later extend $\lamsv$ to $\lamsvplus$ with
dependent types as well as polymorphic types, greatly facilitating the use
of views and viewtypes in programming.

The syntax of $\lamsv$ is given in Figure~\ref{figure:syntax}.
We use $V$ for views and $L$ for addresses.  We use $\lc_0,\lc_1,\ldots$ for
infinitely many distinct constant addresses, which one may assume to be
represented as natural numbers. Also, we write $l$ for a constant address.
We use $\pv$ for proof variables and $\pt$ for proof terms. For each
address $l$, $\pl$ is a constant proof term, whose meaning is to become
clear soon. We use $\Pi$ for a proof variable context, which assigns views
to proof variables.

We use $T$ and $\VT$ for types and viewtypes, respectively. Note that a
type $T$ is just a special form of viewtype. We use $\dt$ for dynamic terms
(that is, programs) and $v$ for values. We write $\pDelta$ ($\eDelta$) for
an intuitionistic (a linear) dynamic variable context, which assign types
(viewtypes) to dynamic variables, and $\Delta$ for a (combined) dynamic
context of the form $(\pDelta;\eDelta)$.  Given $\Delta=(\pDelta;\eDelta)$,
we may use $\Delta,x:\VT$ for $(\pDelta;\eDelta,x:\VT)$; in case the
viewtype $\VT$ is actually a type, we may also use $\Delta,x:\VT$ for
$(\pDelta,x:\VT;\eDelta)$. In addition, given
$\Delta_1=(\pDelta;\eDelta_1)$ and $\Delta_2=(\pDelta;\eDelta_2)$, we write
$\Delta_1\uplus\Delta_2$ for $(\pDelta;\eDelta_1,\eDelta_2)$.

We use $x$ and $f$ for dynamic {\bf lam}-variables and {\bf fix}-variables,
respectively, and $\xf$ for either $x$ or $f$; a {\bf lam}-variable is a
value while a {\bf fix}-variable is not.  We use $c$ for a dynamic
constant, which is either a function $\dcf$ or a constructor $\dcc$.  Each
constant is given a constant type (or c-type) of the form
$(T_1,\ldots,T_n)\Timp T$, where $n$ is the arity of $c$.  We may write
$\dcc$ for $\dcc()$. For instance, each address $l$ is given the c-type
$()\Timp\tptr(l)$; each boolean value is given the c-type $()\Timp\tBool$;
each integer is given the c-type $()\Timp\tInt$; the equality function on
integers can be given the c-type: $(\tInt,\tInt)\Timp\tBool$.  Also, we use
$\tuple{}$ for the unit and $\tunit$ for the unit type.

We use $\mu$ and $\ST$ for state types and states, respectively.  A state
type maps constant addresses to types while a state maps constant addresses
to values. We use $[]$ for the empty mapping and $\mu[l\mapsto T]$ for the
mapping that extends $\mu$ with a link from $l$ to $T$. It is implicitly
assumed that $l$ is not in the domain $\dom(\mu)$ of $\mu$ when
$\mu[l\mapsto T]$ is formed.  Given two state types $\mu_1$ and $\mu_2$
with disjoint domains, we write $\mu_1\otimes\mu_2$ for the standard union
of $\mu_1$ and $\mu_2$. Similar notations are also applicable to states.
In addition, we write $\temd\ST:\mu$ to mean that $\ST(l)$ can be assigned
the type $\mu(l)$ for each $l\in\dom(\ST)=\dom(\mu)$.

We now present some intuitive explanation for certain unfamiliar forms of
types and viewtypes. In $\lamsv$, types are just a special form of
viewtypes. If a dynamic value $v$ is assigned a type, then it consumes no
resources to construct $v$ and thus $v$ can be duplicated.  For instance,
an integer constant $i$ is a value.\footnote{ In particular, we emphasize
that there are simply no ``linear'' integer values in $\lamsv$.}  On the
other hand, if a value $v$ is assigned a viewtype, then it may consume some
resources to construct $v$ and thus $v$ is not allowed to be duplicated.
For instance, the value $(\pl,l)$ can be assigned the viewtype
$(\tInt@l)\Vand\tptr(l)$, which is essentially for a pointer pointing to an
integer; this value contains the resource $\pl$ and thus cannot be
duplicated.
\begin{itemize}
\item The difference between $V\Vimpzero\VT$ and $V\Vimp\VT$
is that the former is a viewtype but not a type while the latter is a type
(and thus a viewtype as well). For instance, the following type
\[\begin{array}{c}
(T_1@L_1)\Vand((T_2@L_2)\Vand(\tptr(L_1)*\tptr(L_2)))\timp
(T_1@L_2)\Vand((T_1@L_2)\Vand\tunit)
\end{array}\]
can be assigned to the function that swaps the contents stored at $L_1$ and
$L_2$. This type is essentially equivalent to the following one:
\[\begin{array}{c}
(T_1@L_1)\Vimp((T_2@L_2)\Vimpzero(\tptr(L_1)*\tptr(L_2)\timp
(T_1@L_2)\Vand((T_1@L_2)\Vand\tunit)))
\end{array}\]
where both $\Vimpzero$ and $\Vimp$ are involved.

\item The difference between $\VT\timpzero\VT$ and $\VT\timp\VT$
is rather similar to that between $V\Vimpzero\VT$ and $V\Vimp\VT$.

\item
In the current implementation of $ATS$, viewtypes of either the form
$V\Vimpzero T$ or the form $\VT\timpzero\VT$ are not directly supported
(though they may be in the future). However, as far as formalization of
viewtypes is concerned, we feel that eliminating such viewtypes would seem
rather {\em ad hoc}.

\end{itemize}

\begin{figure}
\[\begin{array}{c}
\infer[\mbox{\bf(vw-addr)}]
      {\emptyset\tpjg_{[l\mapsto T]} \pl:T@l}{}
\kern18pt
\infer[\mbox{\bf(vw-var)}]
      {\emptyset,\pv:V\tpjg_{[]} \pv:V}{} \\[4pt]

\infer[\mbox{\bf(vw-unit)}]
      {\emptyset\tpjg_{[]} \tuple{}:\prunit}{}
\kern18pt
\infer[\mbox{\bf(vw-tup)}]
      {\Pi_1,\Pi_2\tpjg_{\mu_1\otimes\mu_2} \tuple{\pt_1,\pt_2}:V_1\otimes V_2}
      {\Pi_1\tpjg_{\mu_1} \pt_1: V_1 & \Pi_2\tpjg_{\mu_2} \pt_2: V_2} \\[4pt]

\infer[\mbox{\bf(vw-let)}]
      {\Pi_1,\Pi_2\tpjg_{\mu_1\otimes\mu_2} \letin{\tuple{\pv_1,\pv_2}=\pt_1}{\pt_2}:V}
      {\Pi_1\tpjg_{\mu_1} \pt_1: V_1\otimes V_2 & \Pi_2,\pv_1:V_1,\pv_2:V_2\tpjg_{\mu_2} \pt_2: V} \\[4pt]

\infer[\mbox{\bf(vw-lam)}]
      {\Pi\tpjg_{\mu}\lambda\pv.\pt:V_1\vimp V_2}
      {\Pi,\pv:V_1\tpjg_{\mu} \pt: V_2}
\kern18pt
\infer[\mbox{\bf(vw-app)}]
      {\Pi_1,\Pi_2\tpjg_{\mu_1\otimes\mu_2} \pt_1(\pt_2):V_2}
      {\Pi_1\tpjg_{\mu_1} \pt_1: V_1\vimp V_2 & \Pi_2\tpjg_{\mu_2} \pt_2: V_1} \\[4pt]
\end{array}\]
\caption{The rules for assigning views to proof terms}
\label{figure:view_rules}
\end{figure}
The rules for assigning views to proofs are given in
Figure~\ref{figure:view_rules}. So far only logic constructs in the
multiplicative fragment of intuitionistic linear logic are involved in
forming views, and we plan to handle logic constructs in the additive
fragment of intuitionistic linear logic in future if such a need, which we
have yet to encounter in practice, occurs.  A judgment of the form
$\Pi\tpjg_{\mu}\pt:V$ means that $\pt$ can be assigned the view $V$ if the
variables and constants in $\pt$ are assigned views according to $\Pi$ and
$\mu$, respectively.

\begin{figure}[thp]
\[\begin{array}{c}
\infer[\mbox{\bf(ty-var)}]
      {\emptyset;(\pDelta;\emptyset),x:\VT\tpjg_{[]} x:\VT}
      {} \\[4pt]

\infer[\mbox{\bf(ty-cst)}]
      {\Pi_1,\ldots,\Pi_n;\Delta_1\uplus\ldots\uplus\Delta_n\tpjg_{\mu_1\otimes\ldots\otimes\mu_n} \dc(t_1,\ldots,t_n):T}
      {\tpjg c:(T_1,\ldots,T_n)\Timp T &
       \Pi_i;\Delta_i\tpjg_{\mu_i} t_i:T_i~~\mbox{for $1\leq i\leq n$}} \\[4pt]

\infer[\mbox{\bf(ty-if)}]
      {\Pi_1,\Pi_2;\Delta_1\uplus\Delta_2\tpjg_{\mu_1\otimes\mu_2}\cond(\dt_1,\dt_2,\dt_3):\VT}
      {\Pi_1;\Delta_1\tpjg_{\mu_1}\dt_1:\tBool &
       \Pi_2;\Delta_2\tpjg_{\mu_2}\dt_2:\VT &
       \Pi_2;\Delta_2\tpjg_{\mu_2}\dt_3:\VT } \\[4pt]

\infer[\mbox{\bf(ty-vtup)}]
      {\Pi_1,\Pi_2;\Delta\tpjg_{\mu_1\otimes\mu_2} \tuple{\pt,\dt}:V\Vand\VT}
      {\Pi_1\tpjg_{\mu_1} \pt:V & \Pi_2;\Delta\tpjg_{\mu_2} \dt:\VT} \\[4pt]

\infer[\mbox{\bf(ty-unit)}]
      {\emptyset;(\pDelta;\emptyset)\tpjg_{[]} \tuple{}:\tunit}
      {}
\kern18pt
\infer[\mbox{\bf(ty-tup)}]
      {\Pi_1,\Pi_2;\Delta_1\uplus\Delta_2\tpjg_{\mu_1\otimes\mu_2} \tuple{\dt_1,\dt_2}:\VT_1 * \VT_2}
      {\Pi_1;\Delta_1\tpjg_{\mu_1}\dt_1:\VT_1 &
       \Pi_2;\Delta_2\tpjg_{\mu_2}\dt_2:\VT_2} \\[4pt]

\infer[\mbox{\bf(ty-let)}]
      {\Pi_1,\Pi_2;\Delta_1\uplus\Delta_2\tpjg_{\mu_1\otimes\mu_2}\letin{\tuple{x_1,x_2}=\dt_1}{\dt_2}:\VT}
      {\Pi_1;\Delta_1\tpjg_{\mu_1} \dt_1:\VT_1 * \VT_2 &
       \Pi_2;\Delta_2,x_1:\VT_1,x_2:\VT_2\tpjg_{\mu_2} \dt_2:\VT} \\[4pt]

\infer[\mbox{\bf(ty-vlam0)}]
      {\Pi;\Delta\tpjg_{\mu}\lambda\pv.v:V\Vimpzero\VT}
      {\Pi,\pv:V;\Delta\tpjg_{\mu} v:\VT}
\kern18pt
\infer[\mbox{\bf(ty-vlam)}]
      {\emptyset;(\pDelta;\emptyset)\tpjg_{[]} \lambda\pv.v:V\Vimp\VT}
      {\emptyset,\pv:V;(\pDelta;\emptyset)\tpjg_{[]} v:\VT} \\[4pt]

\infer[\mbox{\bf(ty-vapp)}]
      {\Pi_1,\Pi_2;\Delta\tpjg_{\mu_1\otimes\mu_2} \dt(\pt):\VT}
      {\Pi_1;\Delta\tpjg_{\mu_1} \dt:V\Vimpboth\VT & \Pi_2\tpjg_{\mu_2} \pt:V}  \\[4pt]

\infer[\mbox{\bf(ty-lam0)}]
      {\Pi;\Delta\tpjg_{\mu}\lam{x}{\dt}:\VT_1 \timpzero \VT_2}
      {\Pi;\Delta,x:\VT_1\tpjg_{\mu} \dt:\VT_2}
\kern18pt
\infer[\mbox{\bf(ty-lam)}]
      {\emptyset;(\pDelta;\emptyset)\tpjg_{[]}\lam{x}{\dt}:\VT_1\timp\VT_2}
      {\emptyset;(\pDelta;\emptyset),x:\VT_1\tpjg_{[]} \dt:\VT_2} \\[4pt]

\infer[\mbox{\bf(ty-app)}]
      {\Pi_1,\Pi_2;\Delta_1\uplus\Delta_2\tpjg_{\mu_1\otimes\mu_2} \app(\dt_1,\dt_2):\VT_2}
      {\Pi_1;\Delta_1\tpjg_{\mu_1} \dt_1:\VT_1\timpboth\VT_2 &
       \Pi_2;\Delta_2\tpjg_{\mu_2} \dt_2:\VT_1}  \\[4pt]

\infer[\mbox{\bf(ty-vlet)}]
      {\Pi_1,\Pi_2;\Delta_1\uplus\Delta_2\tpjg_{\mu_1\otimes\mu_2}\letin{\pv\Vand x=\dt_1}{\dt_2}:\VT_2}
      {\Pi_1;\Delta_1\tpjg_{\mu_1} \dt_1:V\Vand \VT_1 &
       \Pi_2,\pv:V;\Delta_2,x:\VT_1\tpjg_{\mu_2} \dt_2:\VT_2} \\[4pt]

\infer[\mbox{\bf(ty-fix)}]
      {\emptyset;(\pDelta;\emptyset)\tpjg_{[]} \fix{f}{\dt}:T}
      {\emptyset;(\pDelta,f:T;\emptyset)\tpjg_{[]} \dt:T} \\[4pt]

\infer[\mbox{\bf(ty-read)}]
      {\Pi_1,\Pi_2;\Delta\tpjg_{\mu_1\otimes\mu_2}\;\fread(\pt,\dt):(T@L)\Vand T}
      {\Pi_1\tpjg_{\mu_1} \pt: T@L &
       \Pi_2;\Delta\tpjg_{\mu_2} \dt:\tptr(L)} \\[4pt]

\infer[\mbox{\bf(ty-write)}]
      {\Pi_1,\Pi_2,\Pi_3;\Delta\tpjg_{\mu_1\otimes\mu_2\otimes\mu_3}\;\fwrite(\pt,\dt_1,\dt_2):(T'@L)\Vand\tunit}
      {\Pi_1\tpjg_{\mu_1} \pt: T@L &
       \Pi_2;\Delta\tpjg_{\mu_2}\dt_1:\tptr (L) &
       \Pi_3\tpjg_{\mu_3} \dt_2: T'} \\[4pt]
\end{array}\]
\caption{The rules for assigning viewtypes to dynamic terms}
\label{figure:viewtype_rules}
\end{figure}
The rules for assigning viewtypes (which include types) are given in
Figure~\ref{figure:viewtype_rules}. We use $\Vimpboth$ for $\Vimp$ or
$\Vimpzero$ in the rule $\mbox{\bf(ty-vapp)}$ and $\timpboth$ for $\timp$
or $\timpzero$ in the rule $\mbox{\bf(ty-app)}$.  Intuitively, a type of
the form $V\Vimpboth\VT$ is for functions from proofs of view $V$ to values
of viewtype $\VT$. Similarly, a type of the form $\VT_1\timpboth\VT_2$ is
for functions from values of viewtype $\VT_1$ to values of viewtype
$\VT_2$.

\begin{proposition}\label{proposition:value}
Assume that $\emptyset;(\emptyset;\emptyset)\tpjg_{\mu} v:T$ is
derivable. Then the state type $\mu$ must equal $[]$.
\end{proposition}
\begin{proof}
By a careful inspection of the rules in Figure~\ref{figure:viewtype_rules}.
\end{proof}
Though simple, Proposition~\ref{proposition:value} is of great importance.
Intuitively, the proposition states that if a closed value is assigned a
type $T$, then the value can be constructed without consuming resources (in
the sense of proof constants $\pl$) and thus is allowed to be duplicated.

We use a judgment of the form $\ST\temd V$ to mean that the state $\ST$
entails the view $V$. The rules for deriving such judgments are given
below:
\[\begin{array}{c}
\infer[]
      {[l\mapsto v]\temd T@l}
      {\emptyset;(\emptyset;\emptyset)\tpjg_{[]} v:T}
\kern10pt
\infer[]
      {\ST_1\otimes\ST_2\temd V_1\otimes V_2}
      {\ST_1\temd V_1 & \ST_2\temd V_2}
\kern10pt
\infer[]
      {\ST\temd V_1\vimp V_2}
      {\ST_0\otimes\ST\temd V_2~\mbox{for each $\ST_0\temd V_1$}} \\[4pt]
\end{array}\]

\begin{lemma}\label{lemma:viewing}
Assume $\Pi\tpjg_{\mu} \pt:V$ is derivable for
$\Pi=\pv_1:V_1,\ldots,\pv_n:V_n$.  If
$\ST=\ST_0\otimes\ST_1\otimes\ldots\otimes\ST_n$ and $\temd\ST_0:\mu$ and $\ST_i\temd
V_i$ for $1\leq n$, then $\ST\temd V$ holds.
\end{lemma}
\begin{proof}
By structural induction on the derivation $\D$ of $\Pi\tpjg_{\mu} \pt:V$.
\end{proof}

We use $[\pv_1,\ldots,\pv_n\mapsto\pt_1,\ldots,\pt_n]$ for a substitution
that maps $\pv_i$ to $\pt_i$ for $1\leq i\leq n$.  Similarly, we use
$[\xf_1,\ldots,\xf_n\mapsto\dt_1,\ldots,\dt_n]$ for a substitution that maps
$\xf_i$ to $\dt_i$ for $1\leq i\leq n$.
\begin{definition}
We define redexes as follows:
\begin{enumerate}
\item
$\letin{\tuple{\pv,x}=\tuple{\pt,v}}{\dt}$ is a redex, and its
reduction is $\subst{v}{x}{\subst{\pt}{\pv}{\dt}}$.
\item
$(\lambda\pv.v)(\pt)$ is a redex, and its reduction is
$\subst{\pt}{\pv}{v}$.
\item
$\letin{\tuple{x_1,x_2}=\tuple{v_1,v_2}}{\dt}$ is a redex, and its
reduction is $\subst{v_1,v_2}{x_1,x_2}{\dt}$.
\item
$\app(\lam{x}{\dt},v)$ is a redex, and its reduction is
$\subst{v}{x}{\dt}$.
\item
$\fix{f}{\dt}$ is a redex, and its reduction is
$\subst{\fix{f}{\dt}}{f}{\dt}$.
\end{enumerate}
\end{definition}
We use $E$ for evaluation contexts, which are defined as follows:
\[\begin{array}{lrcl}
\mbox{evaluation context} & E & ::= &
[] \mid \dc(v_1, \ldots,v_{i-1},E,\dt_{i+1},\ldots,\dt_n) \mid \\
&&&\fread (\pt, E) \mid \fwrite (\pt, E, \dt) \mid \fwrite (\pt, v, E) \mid \\
&&&\pt\Vand E \mid \letin{\pv\Vand x=E}{\dt} \mid E(\pt) \mid \\
&&&\tuple{E,\dt} \mid \tuple{v,E} \mid \letin{\tuple{x_1,x_2}=E}{t} \mid \app(E,\dt)\mid \app(v,E) \\
\end{array}\]
Given $E$ and $\dt$, we write $E[\dt]$ for the dynamic term obtained from
replacing the hole [] in $E$ with $\dt$. Note that such a replacement can
never cause a free variable to be captured.  Given $\ST_1,\ST_2$ and
$\dt_1,\dt_2$, we write $(\ST_1,\dt_1)\steval (\ST_2,\dt_2)$ if
\begin{enumerate}
\item
$\dt_1=E[\dt]$ and $\dt_2=E[\dt']$ for some redex $\dt$ and its reduction, or
\item
$\dt_1=E[\fread(\pt,l)]$ for some $l\in\dom(\ST_1)$ and
$\dt_2=E[\tuple{\pt,\ST_1(l)}]$ and $\ST_2=\ST_1$, or
\item
$\dt_1=E[\fwrite(\pt,l,v)]$ for some $l\in\dom(\ST_1)$ and
$\dt_2=E[\tuple{\pt,\tuple{}}]$ and $\ST_2=\ST_1[l:=v]$.
\end{enumerate}
We write $\ST[l:=v]$ for a state that maps $l$ to $v$ and coincides with
$\ST$ elsewhere. Note that we implicitly assume $l\in\dom(\ST)$ when
writing $\ST[l:=v]$.

We now state some lemmas and theorems involved in establishing the
soundness of $\lamsv$. Please see~\cite{VsTsVTs} for details on their
proofs.
\begin{lemma}[Substitution]
We have the following:
\begin{enumerate}
\item
Assume that both $\Pi_1\tpjg_{\mu_1}\pt_1:V_1$ and
$\Pi_2,\pv:V_1\tpjg_{\mu_2}\pt_2:V_2$ are derivable.  Then
$\Pi_2\tpjg_{\mu_1\otimes\mu_2}\subst{\pt_1}{\pv}{\pt_2}:V_2$ is also
derivable.
\item
Assume that both $\Pi_1\tpjg_{\mu_1}\pt:V$ and
$\Pi_2,\pv:V;\Delta\tpjg_{\mu_2}\dt:\VT$ are derivable.
Then $\Pi_1,\Pi_2;\Delta\tpjg_{\mu_1\otimes\mu_2}\subst{\pt}{\pv}{\dt}:\VT$
is also derivable.
\item
Assume that both $\emptyset;(\emptyset;\emptyset)\tpjg_{\mu_1} v:\VT_1$ and
$\Pi;\Delta,x:\VT_1\tpjg_{\mu_2}\dt:\VT_2$ are derivable.
Then $\Pi;\Delta\tpjg_{\mu_1\otimes\mu_2}\subst{v}{x}{\dt}:\VT_2$
is also derivable.
\end{enumerate}
\end{lemma}
\begin{proof}
By standard structural induction.  In particular, we encounter a need for
Proposition~\ref{proposition:value} when proving (3).
\end{proof}
As usual, the soundness of the type system of $\lamsv$ is built on top of
the following two theorems:
\begin{theorem}[Subject Reduction]\label{theorem:subject_reduction}
Assume $\emptyset;(\emptyset;\emptyset)\tpjg_{\mu_1}\dt_1:\VT$ is derivable
and $\temd\ST_1:\mu_1$ holds.  If $(\ST_1,\dt_1)\steval (\ST_2,\dt_2)$, then
$\emptyset;(\emptyset;\emptyset)\tpjg_{\mu_2}\dt_2:\VT$ is derivable for
some store type $\mu_2$ such that $\temd\ST_2:\mu_2$ holds.
\end{theorem}
\begin{proof}
By structural induction on the derivation of
$\emptyset;(\emptyset;\emptyset)\tpjg_{\mu_1}\dt_1:\VT$.
\end{proof}

\begin{theorem}[Progress]\label{theorem:progress}
Assume that $\emptyset;(\emptyset;\emptyset)\tpjg_{\mu}\dt:\VT$ is
derivable and $\temd\ST:\mu$ holds.  Then either $\dt$ is a value or
$(\ST,\dt)\steval(\ST',\dt')$ for some $\ST'$ and $\dt'$ or $\dt$ is of the
form $E[\dcf(v_1,\ldots,v_n)]$ such that $\dcf(v_1,\ldots,v_n)$ is
undefined.
\end{theorem}
\begin{proof}
By structural induction on the derivation $\D$ of
$\emptyset;(\emptyset;\emptyset)\tpjg_{\mu}\dt:\VT$.
Lemma~\ref{lemma:viewing} is needed when we handling the rules
$\mbox{\bf(ty-read)}$ and $\mbox{\bf(ty-write)}$.
\end{proof}
By Theorem~\ref{theorem:subject_reduction} and
Theorem~\ref{theorem:progress}, we can readily infer that if
$\emptyset;(\emptyset;\emptyset)\tpjg_{\mu}\dt:\VT$ is derivable and
$\temd\ST:\mu$ holds, then either the evaluation of $(\ST,\dt)$ reaches
$(\ST',v)$ for some state $\ST'$ and value $v$ or it continues forever.

Clearly, we can define a function $\erase{\cdot}$ that erases all the proof
terms in a given dynamic term.  For instance, some key cases in the
definition of the erasure function are given as follows:
\[\begin{array}{rcl}
\erase{\letin{\pv\Vand x=t_1}{t_2}} & = & \letin{x=\erase{t_1}}{t_2} \\
\erase{\pt\Vand\dt} & = & \erase{\dt} \\
\erase{\lambda\pv.v} & = & \erase{v} \\
\erase{\dt(\pt)} & = & \erase{\dt} \\
\erase{\fread(\pt,\dt)} & = & \fread(\erase{\dt}) \\
\erase{\fwrite(\pt,\dt_1,\dt_2)} & = & \fwrite(\erase{\dt_1},\erase{\dt_2}) \\
\end{array}\]
It is then straightforward to show that a dynamic term evaluates to a value
if and only if the erasure of the dynamic term evaluates to the erasure of
the value. Thus, there is no need to keep proof terms at run-time: They are
only needed for the purpose of type-checking.
Please see~\cite{CPwTP} for more details on the issue of proof erasure.

\begin{figure}
\[\begin{array}{lrcl}
\mbox{sorts} & \sigma & ::= & \saddr \mid \sbool \mid \sint \mid
\sview \mid \stype \mid \sviewtype \\
\mbox{static contexts} & \Sigma & ::= & \emptysctx \mid \Sigma, a:\sigma \\
\mbox{static addr.} & L & ::= & a \mid l \mid L+I \\
\mbox{static int.} & I & ::= & a \mid i \mid c_I (s_1,\ldots, s_n) \\
\mbox{static prop.} & B & ::= & a \mid b \mid c_B (s_1, \ldots, s_n) \mid
\neg{B} \mid B_1\land B_2 \mid B_1\lor B_2 \mid B_1\pimp B_2\\
\mbox{views} & V & ::= & \ldots \mid B\Bimp V \mid \forall a:\sigma.V \mid
B\Band V \mid \exists a:\sigma.V \\
\mbox{types} & T & ::= & \ldots \mid a \mid \tbool(B) \mid \tint(I) \mid B\Bimp T \mid \forall a:\sigma.T \mid
B\Band T \mid \exists a:\sigma.T \\
\mbox{viewtypes} & \VT & ::= & \ldots \mid B\Bimp\VT \mid \forall
a:\sigma.\VT \mid B\Band\VT \mid \exists a:\sigma.\VT \\
\end{array}\]
\caption{The syntax for the statics of $\lamsvplus$}
\label{figure:lamsvplus_statics}
\end{figure}
\section{Extension}\label{section:extension}
While it supports both views and viewtypes, $\lamsv$ is essentially based
on the simply typed language calculus. This makes it difficult to truly
reap the benefits of views and viewtypes.  In this section, we outline an
extension from $\lamsv$ to $\lamsvplus$ to include universally quantified
types as well as existentially quantified types, greatly facilitating the
use of views and viewtypes in programming. For brevity, most of technical
details are suppressed in this presentation, which is primarily for the
reader to relate the concrete syntax in the examples we present to some
form of formal syntax.

Like an applied type system~\cite{ATS-types03}, $\lamsvplus$ consists of
a static component (statics) and a dynamic component (dynamics).  The
syntax for the statics of $\lamsvplus$ is given in
Figure~\ref{figure:lamsvplus_statics}.  The statics itself is a simply
typed language and a type in it is called a {\em sort}. We assume the
existence of the following basic sorts in $\lamsvplus$:
$\saddr,\sbool,\sint$, $\stype$, $\sview$ and $\sviewtype$; $\saddr$ is the
sort for addresses, and $\sbool$ is the sort for boolean constants, and
$\sint$ is the sort for integers, and $\stype$ is the sort for types, and
$\sview$ is the sort for views, and $\sviewtype$ is the sort for viewtypes.
We use $a$ for static variables, $l$ for address constants
$\lc_0,\lc_1,\ldots$, $b$ for boolean values $\ttrue$ and $\ffalse$, and
$i$ for integers $0,-1,1,\ldots$. We may also use $\sccnull$ for the null
address $\lc_0$.  A term $s$ in the statics is called a static term, and we
use $\Sigma\tpjg s:\sigma$ to mean that $s$ can be assigned the sort
$\sigma$ under $\Sigma$. The rules for assigning sorts to static terms are
all omitted as they are completely standard.

We may also use $L,B,I,T,V,\VT$ for static terms of sorts $\saddr$,
$\sbool$, $\sint$, $\stype$, $\sview$, and $\sviewtype$, respectively. We
assume some primitive functions $\cI$ when forming static terms of sort
$\sint$; for instance, we can form terms such as $I_1 + I_2$, $I_1 - I_2$,
$I_1 * I_2$ and $I_1/I_2$. Also we assume certain primitive functions $\cB$
when forming static terms of sort $\sbool$; for instance, we can form
propositions such as $I_1\leq I_2$ and $I_1\geq I_2$, and for each sort
$\sigma$ we can form a proposition $s_1=_{\sigma} s_2$ if $s_1$ and $s_2$
are static terms of sort $\sigma$; we may omit the subscript $\sigma$ in
$=_{\sigma}$ if it can be readily inferred from the context.  In addition,
given $L$ and $I$, we can form an address $L+I$, which equals $\lc_{n+i}$
if $L=\lc_n$ and $I=i$ and $n+i\geq 0$.

We use $\vB$ for a sequence of propositions and $\Sigma;\vB\temd B$ for a
constraint that means for any $\Theta:\Sigma$, if each proposition in
$\vB[\Theta]$ holds then so does $B[\Theta]$.

In addition, we introduce two type constructors $\tbool$ and $\tint$; given
a proposition $B$, $\tbool(B)$ is the singleton type in which the only
value is the truth value of $B$; similarly, given an integer $I$,
$\tint(I)$ is the singleton type in which the only value is the integer
$I$. Obviously, the previous types $\tBool$ and $\tInt$ can now be defined
as $\exists a:\sbool.\tbool(a)$ and $\exists a:\sint.\tint(a)$, respectively.

\begin{figure}[thp]
\[\begin{array}{lrcl}
\mbox{proof terms} & \pt & ::= &
\ldots \mid\;\iguard{\pt} \mid\;\eguard{\pt} \mid \iforall(\pt) \mid
\eforall(\pt) \mid \\
&&&\Band(\pt) \mid \letin{\Band(\pv)=\pt_1}{\pt_2} \mid
\exists(\pt) \mid \letin{\exists(\pv)=\pt_1}{\pt_2} \\
\mbox{dynamic terms} & \dt & ::= &
\ldots \mid \letin{\Band(\pv)=\pt}{\dt} \mid \letin{\exists(\pv)=\pt}{\dt} \mid \\
&&&\iguard{v} \mid\;\eguard{\dt} \mid \iforall(v) \mid
\eforall(\dt) \mid \\
&&&\Band(\dt) \mid \letin{\Band(x)=\dt_1}{\dt_2} \mid
\exists(\dt) \mid \letin{\exists(x)=\dt_1}{\dt_2} \\
\mbox{values} & v & ::= &
\ldots \mid\;\iguard{v} \mid \iforall(v) \mid \Band(v) \mid \exists(v) \\
\end{array}\]
\caption{The syntax for the dynamics of $\lamsvplus$}
\label{figure:lamsvplus_dynamics}
\end{figure}
Some (additional) syntax for the dynamics of $\lamsvplus$ is given in
Figure~\ref{figure:lamsvplus_dynamics}.  The markers $\iguard{\cdot}$,
$\eguard{\cdot}$, $\iforall(\cdot)$, $\eforall(\cdot)$, $\Band(\cdot)$ and
$\exists(\cdot)$ are primarily introduced to prove the soundness of the
type system of $\lamsvplus$, and please see~\cite{ATS-types03} for
explanation.

We can now also introduce the (built-in) memory access functions
$\dcfgetPtr$ and $\dcfsetPtr$\footnote{The type assigned to $\dcfsetPtr$ is
slightly different from the one in Section~\ref{section:introduction}.}  as
well as the (built-in) memory allocation/deallocation functions $\falloc$
and $\ffree$ and assign them the following constant types:
\[\begin{array}{rcl}
\dcfgetPtr & : &
\forall\lambda.\forall\tau.\tau@\lambda\Vand\tptr(\lambda) \Timp \tau@\lambda\Vand\tau \\
\dcfsetPtr & : &
\forall\lambda.\forall\tau.\ttop@\lambda\Vand\tptr(\lambda)*\tau\Timp\tau@\lambda\Vand\tunit\\
\falloc &~~:~~& \forall\iota.\iota\geq 0\Bimp
(\tint(\iota)\Timp\exists\lambda.\lambda\not=\sccnull\Band(\arrayView(\tunit,\iota,\lambda)\Vand\tptr(\lambda))\\
\ffree &~~:~~& \forall\tau.\forall\iota.\iota\geq 0\Bimp(\arrayView(\tau,\iota,\lambda)\Vand(\tptr(\lambda)*\tint(\iota))\Timp\tunit)\\
\end{array}\]
We use $\ttop$ for the top type such that every type is considered a
subtype of $\ttop$.  When applied to a natural number $n$, $\falloc$
returns a pointer (that is not null) pointing to a newly allocated array of
$n$ units; when applied to a pointer pointing to an array of size $n$,
$\ffree$ frees the array. Note that how these two functions are implemented
is inessential here as long as the implementations meets the constant types
assigned to them.

\begin{figure}[t]
\[\begin{array}{c}
\infer[\mbox{\bf(vw-$\Bimp$+)}]
      {\Sigma;\vB;\Pi\tpjg_{\mu}\iguard{\pt}:B\Bimp V}
      {\Sigma;\vB,B;\Pi\tpjg_{\mu} \pt:V}
\kern18pt
\infer[\mbox{\bf(vw-$\Bimp$-)}]
      {\Sigma;\vB;\Pi\tpjg_{\mu} \eguard{\pt}:V}
      {\Sigma;\vB;\Pi\tpjg_{\mu} \pt:B\Bimp V & \Sigma;\vB\temd B} \\[4pt]

\infer[\mbox{\bf(vw-$\forall$+)}]
      {\Sigma;\vB;\Pi\tpjg_{\mu} \iforall(\pt):\forall a:\sigma.V}
      {\Sigma,a:\sigma;\vB;\Pi\tpjg_{\mu} \pt:V}
\kern18pt
\infer[\mbox{\bf(vw-$\forall$-)}]
      {\Sigma;\vB;\Pi\tpjg_{\mu} \eforall(\pt):\subst{s}{a}{V}}
      {\Sigma;\vB;\Pi\tpjg_{\mu} \pt:\forall a:\sigma.V & \Sigma\tpjg s:\sigma} \\[4pt]

\infer[\mbox{\bf(vw-$\Band$+)}]
      {\Sigma;\vB;\Pi\tpjg_{\mu} \Band(\pt):B\Band V}
      {\Sigma;\vB\temd B & \Sigma;\vB;\Pi\tpjg_{\mu} \pt:V} \\[4pt]

\infer[\mbox{\bf(vw-$\Band$-)}]
      {\Sigma;\vB;\Pi_1,\Pi_2\tpjg_{\mu_1\otimes\mu_2}\letin{\Band(\pv)=\pt_1}{\pt_2}:V_2}
      {\Sigma;\vB;\Pi_1\tpjg_{\mu_1} \pt_1:B\Band V_1 &
       \Sigma;\vB,B;\Pi_2,\pv:V_1\tpjg_{\mu_2} \pt_2:V_2} \\[4pt]

\infer[\mbox{\bf(vw-$\exists$+)}]
      {\Sigma;\vB;\Pi\tpjg_{\mu} \exists(\pt):\exists a:\sigma.V}
      {\Sigma\tpjg s:\sigma & \Sigma;\vB;\Pi\tpjg_{\mu} \pt:\subst{s}{a}{V}} \\[4pt]

\infer[\mbox{\bf(vw-$\exists$-)}]
      {\Sigma;\vB;\Pi_1,\Pi_2\tpjg_{\mu_1\otimes\mu_2}\letin{\exists(\pv)=\pt_1}{\pt_2}:V_2}
      {\Sigma;\vB;\Pi_1\tpjg_{\mu_1} \pt_1:\exists a:\sigma.V_1 &
       \Sigma,a:\sigma;\vB;\Pi_2,\pv:V_1\tpjg_{\mu_2} \pt_2:V_2} \\[4pt]
\end{array}\]
\caption{Some additional rules for assigning views to proof terms}
\label{figure:additional_view_rules}
\end{figure}
\begin{figure}[t]
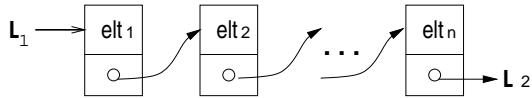

\[\begin{array}{c}
\infer[\mbox{\bf(ty-if)}]
      {\Sigma;\vB;\Pi_1,\Pi_2;\Delta_1\uplus\Delta_2\tpjg_{\mu_1\otimes\mu_2}\cond(\dt_1,\dt_2,\dt_3):\VT}
      {$$\begin{array}{c}
       \Sigma;\vB;\Pi_1;\Delta_1\tpjg_{\mu_1}\dt_1:\tbool (B) \\
       \Sigma;\vB,B;\Pi_2;\Delta_2\tpjg_{\mu_2}\dt_2:\VT \kern12pt
       \Sigma;\vB,\neg{B};\Pi_2;\Delta_2\tpjg_{\mu_2}\dt_3:\VT
       \end{array}$$ } \\[4pt]

\infer[\mbox{\bf(ty-$\Bimp$+)}]
      {\Sigma;\vB;\Pi;\Delta\tpjg_{\mu} \iguard{\dt}:B\Bimp\VT}
      {\Sigma;\vB,B;\Pi;\Delta\tpjg_{\mu} \dt:\VT}
\kern18pt
\infer[\mbox{\bf(ty-$\Bimp$-)}]
      {\Sigma;\vB;\Pi;\Delta\tpjg_{\mu} \eguard{\dt}:\VT}
      {\Sigma;\vB;\Pi;\Delta\tpjg_{\mu} \dt:B\Bimp\VT &
       \Sigma;\vB\temd B} \\[4pt]

\infer[\mbox{\bf(ty-$\forall$+)}]
      {\Sigma;\vB;\Pi;\Delta\tpjg_{\mu} \iforall(\dt):\forall a:\sigma.\VT}
      {\Sigma,a:\sigma;\vB;\Pi;\Delta\tpjg_{\mu} \dt:\VT}
\kern18pt
\infer[\mbox{\bf(ty-$\forall$-)}]
      {\Sigma;\vB;\Pi;\Delta\tpjg_{\mu} \eforall(\dt):\subst{s}{a}{\VT}}
      {\Sigma;\vB;\Pi;\Delta\tpjg_{\mu} \dt:\forall a:\sigma.\VT &
       \Sigma\tpjg s:\sigma} \\[4pt]

\infer[\mbox{\bf(ty-$\Band$+)}]
      {\Sigma;\vB;\Pi;\Delta\tpjg_{\mu} \Band(\dt):B\Band \VT}
      {\Sigma;\vB\temd B & \Sigma;\vB;\Pi;\Delta\tpjg_{\mu} \dt:\VT} \\[4pt]

\infer[\mbox{\bf(ty-$\Band$-)}]
      {\Sigma;\vB;\Pi_1,\Pi_2;\Delta_1\uplus\Delta_2\tpjg_{\mu_1\otimes\mu_2}\letin{\Band(x)=\dt_1}{\dt_2}:\VT_2}
      {\Sigma;\vB;\Pi_1;\Delta_1\tpjg_{\mu_1} \dt_1:B\Band \VT_1 &
       \Sigma;\vB,B;\Pi_2;\Delta_2,x:\VT_1\tpjg_{\mu_2} \dt_2:\VT_2} \\[4pt]

\infer[\mbox{\bf(ty-$\exists$+)}]
      {\Sigma;\vB;\Pi;\Delta\tpjg_{\mu} \exists(\dt):\exists a:\sigma.\VT}
      {\Sigma\tpjg_{\mu} s:\sigma & \Sigma;\vB;\Pi;\Delta\tpjg_{\mu} \dt:\subst{s}{a}{\VT}} \\[4pt]

\infer[\mbox{\bf(ty-$\exists$-)}]
      {\Sigma;\vB;\Pi_1,\Pi_2;\Delta_1\uplus\Delta_2\tpjg_{\mu_1\otimes\mu_2}\letin{\exists(x)=\dt_1}{\dt_2}:\VT_2}
      {\Sigma;\vB;\Pi_1;\Delta_1\tpjg_{\mu_1} \dt_1:\exists a:\sigma.\VT_1 &
       \Sigma,a:\sigma;\vB;\Pi_2;\Delta_2,x:\VT_1\tpjg_{\mu_2} \dt_2:\VT_2} \\[4pt]
\end{array}\]
\caption{Some additional rules for assigning viewtypes to dynamic terms}
\label{figure:additional_viewtype_rules}
\end{figure}
A judgment for assigning a view to a proof is now of the form
$\Sigma;\vB;\Pi\tpjg_{\mu}\pt:V$, and the rules in
Figure~\ref{figure:view_rules} need to be modified properly.  Intuitively,
such a judgment means that $\Pi[\Theta]\tpjg_{\mu}\pt:V[\Theta]$ holds for
any substitution $\Theta:\Sigma$ such that each $B$ in $\vB[\Theta]$
holds. Some additional rules for assigning views to proof terms are given
in Figure~\ref{figure:additional_view_rules}.  Similarly, a judgment for
assigning a viewtype to a dynamic term is now of the form
$\Sigma;\vB;\Pi;\Delta\tpjg\dt:\VT$, and the rules in
Figure~\ref{figure:viewtype_rules} need to be modified properly.  Some
additional rules for assigning viewtypes to dynamic terms are given in
Figure~\ref{figure:additional_viewtype_rules}.

Given the development detailed in~\cite{ATS-types03}, it is a standard
routine to establish the soundness of the type system of $\lamsvplus$. The
challenge here is really not in the proof of the soundness; it is instead
in the formulation of the rules presented in
Figure~\ref{figure:view_rules}, Figure~\ref{figure:additional_view_rules},
Figure~\ref{figure:viewtype_rules} and
Figure~\ref{figure:additional_viewtype_rules} for assigning views and
viewtypes to proof terms and dynamic terms, respectively.  We are now ready
to present some running examples taken from the current implementation of
ATS.

\begin{figure}[t]
\begin{center}
\fontsize{11}{12}\selectfont
\begin{minipage}{15cm}
\begin{verbatim}
fun arrayMap {a1: type, a2: type, n: int, l: addr | n >= 0}
   (pf: arrayView (a1, n, l) | A: ptr l, n: int n, f: a1 -> a2)
  : '(arrayView (a2, n, l) | unit) =
  if n igt 0 then // [igt]: the greater-than function on integers
    let
       prval ArraySome (pf1, pf2) = pf
       val '(pf1 | v) = getPtr (pf1 | A)
       val '(pf1 | _) = setPtr (pf1 | A, f v)
       // [ipred]: the predessor function on integers
       val '(pf2 | _) = arrayMap (pf2 | A + 1, ipred n, f)
    in
       '(ArraySome (pf1, pf2) | '())
    end
  else let prval ArrayNone () = pf in '(ArrayNone () | '()) end
\end{verbatim}
\end{minipage}
\end{center}
\caption{An implementation of in-place array map function}
\label{figure:arrayMap_function}
\end{figure}
\begin{figure}[t]
\fontsize{9}{10}\selectfont
\begin{center}
\begin{minipage}{15cm}
\begin{verbatim}
dataview slsegView (type, int, addr, addr) =
  | {a:type, l:addr} SlsegNone (a, 0, l, l)
  | {a:type, n:int, first:addr, next:addr, last:addr | n >= 0, first <> null}
    // 'first <> null' is added so that nullity test can
    // be used to check whether a list segment is empty. 
    SlsegSome (a, n+1, first, last) of
      ((a, ptr next) @ first, slsegView (a, n, next, last))

viewdef sllistView (a:type, n:int, l:addr) = slsegView (a, n, l, null)

fun reverse {a:type, n:int, l:addr | n >= 0} // in-place singly-linked list reversal
   (pf: sllistView (a, n, l) | p: ptr l)
  : [l: addr] '(sllistView (a, n, l) | ptr l) =
  let
     fun rev {n1:int,n2:int,l1:addr,l2:addr | n1 >= 0, n2 >= 0}
        (pf1: sllistView (a,n1,l1), pf2: sllistView (a,n2,l2) |
         p1: ptr l1, p2: ptr l2)
       : [l:addr] '(sllistView (a, n1+n2, l) | ptr l) =
       if isNull p2 then let prval SlsegNone () = pf2 in '(pf1 | p1) end
       else let
          prval SlsegSome (pf21, pf22) = pf2
          prval '(pf210, pf211) = pf21
          val '(pf211 | next) = getPtr (pf211 | p2 + 1)
          val '(pf211 | _) = setPtr (pf211 | p2 + 1, p1)
          prval pf1 = SlsegSome ('(pf210, pf211), pf1)
       in rev (pf1, pf22 | p2, next) end
  in
     rev (SlsegNone (), pf | null, p)
  end
\end{verbatim}
\end{minipage}
\end{center}
\caption{An implementation of in-place singly-linked list reversal}
\label{figure:sllistReversal}
\end{figure}
A clearly noticeable weakness in many typed programming languages lies in
the treatment of the allocation and initialization of arrays (and many
other data structures).  For instance, the allocation and initialization of
an array in SML is atomic and cannot be done separately. Therefore, copying
an array requires a new array be allocated and then initialized before
copying can actually proceed. Though the initialization of the newly
allocated array is completely useless, it unfortunately cannot be avoided.
In $\lamsvplus$ (extended with recursive stateful views), a function of the
following type can be readily implemented that replaces elements of type
$T_1$ in an array with elements of type $T_2$ when a function of type
$T_1\timp T_2$ is given:
\begin{center}
\[\begin{array}{l}
\forall\tau_1.\forall\tau_2.\forall\iota.\forall\lambda. \\
\kern12pt
\iota\geq 0\Bimp
(\arrayView (\tau_1,\iota,\lambda) \Vand (\tptr(\lambda) * \tint(\iota) * (\tau_1\timp\tau_2)) \timp
\arrayView (\tau_2, \iota, \lambda)\Vand\tptr(\lambda))
\end{array}\]
\end{center}
With such a function, the allocation and initialization of an array can
clearly be separated. In Figure~\ref{figure:arrayMap_function}, we present
an implementation of in-place array map function in ATS.  Note that
$\arrayView$ is declared as a recursive stateful view constructor in
Figure~\ref{figure:arrayView}.\footnote{Though the notion of recursive
stateful view is not present in $\lamsvplus$, it should be understood that
this notion can be readily incorporated.} Note that for a proof $\pf$ of
view $\arrayView(T,I,L)$ for some type $T$, integer $I>0$ and address $L$,
the following syntax in Figure~\ref{figure:arrayMap_function} means that
$\pf$ is decomposed into two proofs $\pf_1$ and $\pf_2$ of views $T@L$ and
$\arrayView(T,I-1,L+1)$, respectively:
\begin{verbatim}
                  prval ArraySome (pf1, pf2) = pf
\end{verbatim}
The rest of the syntax in Figure~\ref{figure:arrayMap_function} should then
be easily accessible.

The next example we present is in Figure~\ref{figure:sllistReversal}, where
a recursive view constructor $\slsegView$ is declared.
Note that we write $(T_0,\ldots,T_n)@L$
for a sequence of views: $T_0@(L+0),\ldots,T_n@(L+n)$. Given a type $T$, an
integer $I$, and two addresses $L_1$ and $L_2$, $\slsegView(T,I,L_1,L_2)$ is a
view for a singly-linked list segment pictured as follows:
\begin{center}
\begin{pspicture}
\rput(0,1.5){\epsfysize=36pt\epsfbox{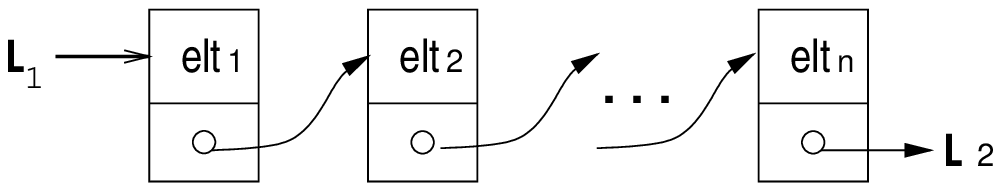}}
\end{pspicture}
\end{center}
such that (1) each element of the segment is of type $T$, and (2) the
length of the segment is $I$, and (3) the segment starts at $L_1$ and ends
at $L_2$.  There are two view proof constructors $\SlsegNone$ and
$\SlsegSome$ associated with $\slsegView$.  A singly-linked list is just a
special case of singly-linked list segment that ends at the null
address. Therefore, $\sllistView(T,I,L)$ is a view for a singly-linked list
pictured as follows:
\begin{center}
\begin{pspicture}
\rput(0.125,1.5){\epsfysize=36pt\epsfbox{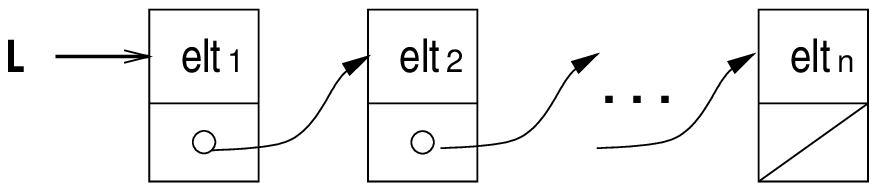}}
\end{pspicture}
\end{center}
such that each element in it is of type $T$ and its length is $I$.  To
demonstrate how such a view can be used in programming, we implement an
in-place reversal function on singly-linked-lists in
Figure~\ref{figure:sllistReversal}, which is given the following type:
\[\begin{array}{l}
\forall\tau.\forall\iota.\forall\lambda.
\iota\geq 0\Bimp(
\sllistView(\tau,\iota,\lambda)\Vand\tptr(\lambda)\timp
\exists\lambda (\sllistView(\tau,\iota,\lambda)\Vand\tptr(\lambda)))
\end{array}\]
indicating that this is a length-preserving function.

\section{Persistent Stateful Views}\label{section:PSV}
There is so far an acute problem with $\lamsvplus$ that we have not
mentioned. Given the linearity of stateful views, we simply can not support
pointer sharing in $\lamsvplus$. For instance, the kind of references in
ML, which require pointer sharing, can not be directly handled.\footnote{
We can certainly add into $\lamsvplus$ some primitives to support
references, but such a solution is inherently {\em ad hoc}.}  This problem
would impose a crippling limitation on stateful views if it could not be
properly resolved.  Fortunately, We have already found a solution to the
problem by introducing a notion of persistent stateful views, and $\atssv$
is essentially the type system that extends $\lamsvplus$ with persistent
stateful views.  For the sake of brevity, we refer the reader
to~\cite{ATSwSV} for the detailed theoretical development of persistent
stateful views.  In the following, we briefly present some simple intuition
behind our solution.

We start with an implementation of references in ATS. Essentially,
references can be regarded as a special form of pointers such that we have
no obligation to provide proofs (of views) when reading from or writing to
them. In ATS, the type constructor $\sccref$ for forming types for
references can be defined as follows:
\begin{verbatim}
// [!] stands for a box
typedef ref (a: type) = [l:addr] '(!(a @ l) | ptr l)
\end{verbatim}
In formal notation, given a type $T$, $\sccref(T)$ is defined to be
$\exists\lambda.(\square(T@\lambda)\mid\tptr(\lambda))$, where
$\square(T@\lambda)$ is a {\em persistent} stateful view. When compared to
the original stateful views, which we now refer to as {\em ephemeral}
stateful views, persistent stateful views are intuitionistic and proofs of
such views can be duplicated.  Given an ephemeral view $V$, we can
construct a persistent view $\square{V}$, which we may refer to as the
boxed $V$. Note that we use \texttt{!} for $\square$ in the concrete
syntax. At this point, we emphasize that it is incorrect to assume that
a persistent view $\square{V}$ implies the ephemeral view $V$.
Essentially, $\square$ acts as a form of modality, which restricts
the use of a boxed view.

\begin{figure}[t]
\fontsize{9}{10}\selectfont
\begin{verbatim}
// [getPtr0] is of the type: {a:type, l:addr} (a @ l | (*none*) | ptr l) -> a
fun getPtr0 {a:type, l:addr} (pf: a @ l | (*none*) | p: ptr l): a =
  getPtr (pf | p)

// [setPtr0] is of the type: {a:type, l:addr} (a @ l | (*none*) | ptr l, a) -> unit
fun setPtr0 {a:type, l:addr} (pf: a @ l | (*none*) | p: ptr l, x: a): unit =
  setPtr (pf | p, x)

fun getRef {a:type} (r: ref a): a =
  let val '(pf | p) = r in getPtr0 (pf | (*none*) | p) end

fun setRef {a:type} (r: ref a, x: a): unit =
  let val '(pf | p) = r in setPtr0 (pf | (*none*) | p, x) end
\end{verbatim}
\caption{Implementing references}
\label{figure:reference}
\end{figure}
Given a view $V$, we say that a function of type $V\Vand T_1\timp V\Vand
T_2$ treats the view $V$ as an invariant since the function consumes a
proof of $V$ and then produces another proof of the same $V$. What we can
formally show is that such a function can also be used as a function of
type $\square{V}\Vand T_1\timp T_2$. For instance, the function
$\dcfgetPtr$, which is given the type
$\forall\tau.\forall\lambda.(\tau@\lambda)\Vand\tptr(\lambda) \timp
(\tau@\lambda)\Vand\tau$, can be used as a function ($\dcfgetRef$) of type
$\forall\tau.\forall\lambda.\square(\tau@\lambda)\Vand\tptr(\lambda) \timp
\tau$ to read from a reference. Similarly, we can form a function
($\dcfsetRef$) of type
$\forall\tau.\forall\lambda.\square(\tau@\lambda)\Vand(\tptr(\lambda)*\tau)
\timp\tunit$ for writing to a reference. The actual implementation of
$\dcfgetRef$ and $\dcfsetRef$ are given in Figure~\ref{figure:reference}.
Note that in the concrete syntax, we write
\texttt{ (V1 | V2 | VT1) -> VT2 } for
\begin{center}
\texttt{ (V1, V2 | VT1) -> '(V1 | VT2) }
\end{center}
so as to indicate that $V_1$ is an invariant. Therefore, $\dcfgetPtrZ$
and $\dcfsetPtrZ$ in Figure~\ref{figure:reference} are declared to be of
the following types:
\[\begin{array}{rcl}
\dcfgetPtrZ &~~:~~&
\forall\tau.\forall\lambda.(\tau@\lambda)\Vand\tptr(\lambda)\timp(\tau@\lambda)\Vand\tau
\\
\dcfsetPtrZ &~~:~~&
\forall\tau.\forall\lambda.(\tau@\lambda)\Vand(\tptr(\lambda) * \tau_2)\timp(\tau@\lambda)\Vand\tunit
\\
\end{array}\]
Next we show that both product and sum types are implementable in ATS
through the use of persistent stateful views.

\begin{figure}[thp]
\fontsize{10}{11}\selectfont
\begin{verbatim}
typedef pair (a1: type, a2: type) = [l: addr] '(!(a1 @ l), !(a2 @ l+1) | ptr l)

fun makePair {a1:type, a2:type} (x1: a1, x2: a2): pair (a1, a2) =
  let
     val '(pf | p) = alloc (2)
     prval ArraySome (pf1, ArraySome (pf2, ArrayNone ())) = pf
     val '(pf1 | _) = setPtr (pf1 | p, x1)
     val '(pf2 | _) = setPtr (pf2 | p + 1, x2)
  in
     '(viewbox pf1, viewbox pf2 | p)
  end

fun getFst {a1:type, a2:type} (p: pair (a1, a2)): a1 =
  let val '(pf1, _ | p0) = p in getPtr0 (pf1 | (*none*) | p0) end

fun getSnd {a1:type, a2:type} (p: pair (a1, a2)): a2 =
  let val '(_, pf2 | p0) = p in getPtr0 (pf2 | (*none*) | p0 + 1) end

typedef sum (a1: type, a2: type) =
  [l: addr, i: int | i == 0 || i == 1]
    '(!(int (i) @ l), {i == 0} !(a1 @ l+1), {i == 1} !(a2 @ l+1) | ptr l)

// left injection
fun inl {a1: type, a2: type} (x: a1): sum (a1, a2) =
  let
     val '(pf | p) = alloc (2)
     prval ArraySome (pf1, ArraySome (pf2, ArrayNone ())) = pf
     val '(pf1 | _) = setPtr (pf1 | p, 0)
     val '(pf2 | _) = setPtr (pf2 | p + 1, x)
  in
     '(viewbox pf1, viewbox pf2, '() | p)
  end

// right injection
fun inr {a1: type, a2: type} (x: a2): sum (a1, a2) =
  let
     val '(pf | p) = alloc (2)
     prval ArraySome (pf1, ArraySome (pf2, ArrayNone ())) = pf
     val '(pf1 | _) = setPtr (pf1 | p, 1)
     val '(pf2 | _) = setPtr (pf2 | p + 1, x)
  in
     '(viewbox pf1, '(), viewbox pf2 | p)
  end

\end{verbatim}
\caption{implementations of product and sum}
\label{figure:product_sum_implementation}
\end{figure}
\subsection{Implementing Product and Sum}
In ATS, both product and sum are implementable in terms of other primitive
constructs.  For instance, product and sum are implemented in
Figure~\ref{figure:product_sum_implementation}.  In the implementation, the
type $\sccpair(T_1,T_2)$ for a pair with the first and second components of
types $T_1$ and $T_2$, respectively, is defined to be:
$$\exists\lambda.(\square(T_1@\lambda) \Band \square(T_2@\lambda+1))\Vand
\tptr(\lambda)$$ The function $\dcfmakePair$ is given the type
$\forall\tau_1.\forall\tau_2.(\tau_1,\tau_2)\timp\sccpair(\tau_1,\tau_2)$,
that is, it takes values of types $T_1$ and $T_2$ to form a value of type
$\sccpair(T_1,T_2)$.\footnote{In ATS, we support functions of multiple
arguments, which should be distinguished from functions that takes a tuple
as a single argument.}.  Note that $\scfviewbox$ is a primitive that turns
an ephemeral stateful view $V$ into a persistent stateful view
$\square{V}$.

The implementation of sum is more interesting. We define $T_1+T_2$ to be
$\exists\iota.\exists\lambda.(\iota=0\lor\iota=1)\Bimp V\Vand\tptr(\lambda)$, where
$V$ is given as follows:
\begin{center}
\[\begin{array}{c}
(\square(\tint(\iota) @ \lambda) \Band (\iota = 0\Bimp\square(T_1 @ \lambda+1)) \Band (\iota = 1\Bimp\square(T_2 @ \lambda+1)))
\end{array}\]
\end{center}
Note the use of guarded persistent stateful views here. Essentially, a
value of type $T_1+T_2$ is represented as a tag (which is an integer of
value $0$ or $1$) followed by a value of type $T_1$ or $T_2$ determined by
the value of the tag.  Both the left and right injections can be
implemented straightforwardly.  Given that recursive types are available in
ATS, datatypes as supported in ML can all be readily implemented in a
manner similar to the implementation of sum.

\section{Current Status of ATS}
We have finished a running implementation of ATS, which is currently
available on-line~\cite{ats-lang}, and the type system presented in this paper
is a large part of the implementation. At this moment, well-typed programs
in ATS are interpreted. We have so far gather some empirical evidence in
support of the practicality of programming with stateful views. For
instance, the library of ATS alone already contains more than 20,000 lines
of code written in ATS itself, involving a variety of data structures such
as cyclic linked lists and doubly-linked binary trees that make
(sophisticated) use of pointers.  In particular, the library code includes
a portion modeled after the Standard Template Library (STL) of
C++~\cite{STL}, and the use of stateful views (both ephemeral and
persistent) is ubiquitous in this portion of code.

\section{Related Work and Conclusion}
A fundamental issue in programming is on program verification, that is,
verifying (in an effective manner) whether a program meets its
specification.  In general, existing approaches to program verification can
be classified into two categories.  In one category, the underlying theme
is to develop a proof theory based Floyd-Hoare logic (or its variants) for
reasoning about imperative stateful programs.  In the other category, the
focus is on developing a type theory that allows the use of types in
capturing program properties.

While Floyd-Hoare logic has been studied for at least three decades
\cite{HOARE69,HOARE71}, its actual use in general software practice is
rare. In the literature, Floyd-Hoare logic is mostly employed to prove the
correctness of some (usually) short but often intricate programs, or to
identify some subtle problems in such programs. In general, it is still as
challenging as it was to support Floyd-Hoare logic in a realistic
programming language.  On the other hand, the use of types in capturing
program invariants is wide spread. For instance, types play a significant
r{\^o}le in many modern programming languages such as ML and Java.
However, we must note that the types in these programming languages are of
relatively limited expressive power when compared to Floyd-Hoare logic. In
Martin-L{\"o}f's constructive type theory~\cite{MARTINLOF84,NORDSTROM90},
dependent types offer a precise means to capture program properties, and
complex specifications can be expressed in terms of dependent types. If
programs can be assigned such dependent types, they are {\em guaranteed} to
meet the specifications.  However, because there exists no separation
between programs and types, that is, programs may be used to construct
types, a language based on Martin-L{\"o}f's type theory is often too pure
and limited to be useful for practical purpose.

In Dependent ML (DML), a restricted form of dependent types is proposed
that completely separates programs from types, this design makes it rather
straightforward to support realistic programming features such as general
recursion and effects in the presence of dependent types. Subsequently,
this restricted form of dependent types is used in designing
Xanadu~\cite{Xanadu} and DTAL~\cite{DTALabs} so as to reap similar benefits
from dependent types in imperative programming.  In hindsight, the type
system of Xanadu can be viewed as an attempt to combine type theory with
Floyd-Hoare logic.

In Xanadu, we follow a strategy in Typed Assembly Language (TAL)~\cite{TAL}
to statically track the changes made to states during program evaluation. A
fundamental limitation we encountered is that this strategy only allows the
types of the values stored at a fixed number of addresses to be tracked in
any given program, making it difficult, if not entirely impossible, to
handle data structures such as linked lists in which there are an
indefinite number of pointers involved. We have seen several attempts made
to address this limitation.  In~\cite{sagiv98solving}, finite shape graphs
are employed to approximate the possible shapes that mutable data
structures (e.g., linked lists) in a program can take on. A related
work~\cite{AliasTypes} introduces the notion of alias types to model
mutable data structures such as linked lists. However, the notion of view
changes in $\atssv$ is not present in these works. For instance, an alias
type can be readily defined for circular lists, but it is rather unclear
how to program with such an alias type. As a comparison, a view can be
defined as follows in $\atssv$ for circular lists of length $n$:
\begin{verbatim}
viewdef circlistView (a:type,n:int,l:addr) = slsegView (a,n,l,l)
\end{verbatim}
With properly defined functions for performing view changes, we can easily
program with circular lists. For instance, we have finished a queue
implementation based on circular lists~\cite{ats-lang}.

Along a related but different line of research, separation
logic~\cite{SeparationLogic} has recently been introduced as an extension
to Hoare logic in support of reasoning on mutable data structures. The
effectiveness of separation logic in establishing program correctness is
convincingly demonstrated in various nontrivial examples (e.g.,
singly-linked lists and doubly-linked lists).  It can be readily noticed
that proofs formulated in separation logic in general correspond to the
functions in $\atssv$ for performing view changes, though a detailed
analysis is yet to be conducted.  In a broad sense, $\atssv$ can be viewed
as a novel attempt to combine type theory with (a form of) separation
logic.  In particular, the treatment of functions as first-class values is
a significant part of $\atssv$, which is not addressed in separation logic.
Also, we are yet to see programming languages (or systems ) that can
effectively support the use of separation logic in practical programming.

There is a large body of research on applying linear type theory based on
linear logic~\cite{GIRARD87} to memory management
(e.g.~\cite{WADLER90,chirimar96reference,turner99operational,kobayashi99quasi,igarashi00garbage,Hofmann2000}),
and the work~\cite{petersen03a} that attempts to give an account for data
layout based on ordered linear logic~\cite{polakow99relating} is closely
related to $\atssv$ in the aspect that memory allocation and data
initialization are completely separated.  However, due to the rather
limited expressiveness of ordered linear logic, it is currently unclear how
recursive data structures such as arrays and linked lists can be properly
handled.

There have been a large number of studies on verifying program safety
properties by tracking state changes. For instance, Cyclone~\cite{Cyclone}
allows the programmer to specify safe stack and region memory allocation;
both CQual~\cite{CQual} and Vault~\cite{Vault} support some form of
resource usage protocol verification; ESC~\cite{ESC} enables the programmer
to state various sorts of program invariants and then employs theorem
proving to prove them; CCured~\cite{CCured} uses program analysis to show
the safety of mostly unannotated C programs.  In particular, the type
system of Vault also rests on (a form of) linear logic, where two
constructs {\em adoption} and {\em focus} are introduced to reduce certain
conflicts between linearity and sharing. Essentially, {\em focus}
temporarily provides a linear view on an object of nonlinear type while
{\em adoption} does the opposite, and our treatment of persistent stateful
views bears some resemblance to this technique. However, the underlying
approaches taken in these mentioned studies are in general rather different
from ours and a direct comparison seems difficult.

In this paper, we are primarily interested in providing a framework based
on type theory to reason about program states. This aspect is also shared
in the research on an effective theory of type
refinements~\cite{mandelbaum+:effref}, where the aim is to develop a
general theory of type refinements for reasoning about program
states. Also, the formalization of $\atssv$ bears considerable resemblance
to the formalization of the type system in~\cite{mandelbaum+:effref}.  This
can also be said about the work in~\cite{LogicalApproachToStackTyping},
where the notion of primitive stateful view like $T@L$ is already present
and there are also various logic connectives for combining primitive
stateful views.  However, the notions such as recursive stateful views
(e.g., $\arrayView$) and view changes, which constitute the key
contributions of this paper, have no counterparts in
either~\cite{mandelbaum+:effref} or~\cite{LogicalApproachToStackTyping}.
Recently, a linear language with locations ($L^3$) is presented by
Morrisett et al~\cite{LLL}, which attempts to explore foundational typing
support for strong updates. In $L^3$, stateful views of the form $T@L$ are
present, but recursive stateful views are yet to be developed.  The notion
of {\em freeze} and {\em thaw} in $L^3$ seems to be closely related to our
handling of persistent stateful views, but we have also noticed some
fundamental differences. For instance, the function $\scfviewbox$ that
turns an ephemeral stateful view into a persistent stateful view seems to
have no counterpart in $L^3$.

Another line of related studies are in the area of shape
analysis~\cite{sagiv98solving,TVLA,ShapeAnalysisViaThreeValuedLogic}. While
we partly share the goal of shape analysis, the approach we take is
radically different from the one underlying shape analysis. Generally
speaking, TVLA performs fixed-point iterations over abstract descriptions
of memory layout.  While it is automatic (after an operational semantics is
specified in 3-valued logic for a collection of primitive operations on the
data structure in question), it may lose precision when performing
fixed-point iteration and thus falsely reject programs. Also, many
properties that can be captured by types in $\atssv$ seem to be beyond the
reach of TVLA. For instance, the type of the list reversal function in
Figure~\ref{figure:sllistReversal} states that it is length-preserving, but
this is difficult to do in TVLA.  Overall, it probably should be said that
$\atssv$ (type theory) and shape analysis (static analysis) are
complementary.

In~\cite{PointerAssertionLogicEngine}, a framework is presented for
verifying partial program specifications in order to capture type and
memory errors as well as to check data structure invariants.  In general, a
data structure can be handled if it can be described in terms of graph
types~\cite{GraphTypes}. Programs are annotated with partial specifications
expressed in Pointer Assertion Logic.  In particular, loop and function
call invariants are required in order to guarantee the decidability of
verification. This design is closely comparable to ours given that
invariants correspond to types and verification corresponds to
type-checking.  However, arithmetic invariants are yet to be supported in
the framework.

In summary, we have presented the design and formalization of a type system
$\atssv$ that makes use of stateful views in capturing invariants in
stateful programs that may involve (sophisticated) pointer manipulation. We
have not only established the type soundness of $\atssv$ but also given a
variety of running examples in support of programming with stateful views.
We are currently keen to continue the effort in building the programming
language ATS, making it suitable for both high-level and low level
programming.  With $\atssv$, we believe that a solid step toward reaching
this goal is made.

\newcommand{\etalchar}[1]{$^{#1}$}

\end{document}

\begin{figure}
\begin{center}
\begin{minipage}{15cm}
\begin{verbatim}
fun swap (pf1: T1 @ L1, pf2: T2 @ L2 | p1: ptr L1, p2: ptr L2)
  : (T1 @ L2, T2 @ L1 | unit) =
  let
     val '(pf1' | v1) = read (pf1 | p1)
     val '(pf2' | v2) = read (pf2 | p2)
     val '(pf1'' | _) = write (pf1' | p1, v)
     val '(pf2'' | _) = write (pf2' | p1, v)
  in
     '(pf1'', pf2'' | '())
  end
\end{verbatim}
\end{minipage}
\end{center}
\end{figure}